\documentclass[11pt, journal, draftclsnofoot, onecolumn]{IEEEtran}

\usepackage{mathtools,enumitem,color,subfigure,bbm}
\usepackage{graphicx,tabularx,array}
\usepackage{hyperref}

\usepackage{amsthm,mathrsfs}
\usepackage{makecell}
\usepackage{multirow}
\usepackage[T1]{fontenc}
\usepackage{cite}
\usepackage{txfonts}
\usepackage{enumerate}

\IEEEoverridecommandlockouts

\usepackage{amsthm}
\usepackage{cite}
\usepackage{enumerate}
\usepackage{enumitem}
\usepackage{url}              
\usepackage{graphicx}         
\usepackage{mathtools}

\usepackage{bm}
\usepackage{bbm}
 \usepackage{amssymb}

\theoremstyle{definition}
\newtheorem{Theorem}{Theorem}
\newtheorem{Proposition}{Proposition}

\newtheorem{Corollary}{Corollary}
\newtheorem{Example}{Example}

\newtheorem{Definition}{Definition}

\begin{document}
\title{Quantum Advantage in Non-Interactive Source Simulation} 




\author{%
  \IEEEauthorblockN{Hojat Allah Salehi\IEEEauthorrefmark{1},
                    Farhad Shirani\IEEEauthorrefmark{1},
                    S. Sandeep Pradhan\IEEEauthorrefmark{2},
                   }
  \\\IEEEauthorblockA{\IEEEauthorrefmark{1}%
 Florida International University, Miami, FL,
                     \{hsalehi,fshirani\}@fiu.edu}
\\  \IEEEauthorblockA{\IEEEauthorrefmark{2}%
                University of Michigan, Ann Arbor, MI, pradhanv@umich.edu}
}

\maketitle


\begin{abstract}
   This work considers the non-interactive source simulation problem (NISS). 
   In the standard NISS scenario, a pair of distributed agents, Alice and Bob, observe a distributed binary memoryless source $(X^d,Y^d)$ generated based on joint distribution $P_{X,Y}$. The agents wish to produce a pair of discrete random variables $(U_d,V_d)$ with joint distribution $P_{U_d,V_d}$, such that $P_{U_d,V_d}$ converges in total variation distance to a target distribution $Q_{U,V}$.
   Two variations of the standard NISS scenario are considered. In the first variation, in addition to $(X^d,Y^d)$ the agents have access to a shared Bell state. The agents each measure their respective state, using a measurement of their choice, and use its classical output along with  $(X^d,Y^d)$ to simulate the target distribution. This scenario is called the entanglement-assisted NISS (EA-NISS). In the second variation, the agents have access to a classical common random bit $Z$, in addition to $(X^d,Y^d)$. This scenario is called the classical common randomness NISS (CR-NISS). It is shown that for binary-output NISS scenarios, the set of feasible distributions for EA-NISS and CR-NISS are equal with each other. Hence, there is not quantum advantage in these EA-NISS scenarios. For non-binary output NISS scenarios, it is shown that the set of simulatable distributions in the CR-NISS scenario has measure zero within the set of distributions in EA-NISS scenario.
   
\end{abstract}

\section{Introduction}
A fundamental problem of interest in information theory and theoretical computer science  is to quantify the correlation between the outputs in distributed processing of random sequences. 
The problem has been extensively studied in the classical settings under the frameworks of non-interactive source simulation (NISS) \cite{ghazi2016decidability,ghazi2017dimension,de2018non,khorasgani2021decidability,bhushan2023secure,yu2022common,sudan2019communication,Shiraniboolean2017,shirani2019sub,shirani2023non}. Applications of such quantification include design of consensus protocols \cite{cachin2000random,rabin1983randomized}, proof-of-stake based blockchain \cite{bentov2016snow,gilad2017algorand}, scaling smart contracts \cite{das2018yoda}, anonymous communication \cite{goel2003herbivore}, private browsing \cite{dingledine2004tor},  publicly auditable auctions and
lottery \cite{bonneau2015bitcoin}, and cryptographic parameter
generation \cite{ahlswede1993common,maurer1993secret,baigneres2015trap}, among others.

In classical NISS, two agents, Alice and Bob, observe a pair of random sequences $X^d$ and $Y^d$, respectively, for some $d\in \mathbb{N}$, where $X^d,Y^d$ are generated independently and based on an identical joint distribution $P_{X,Y}$. 
Their objective is to non-interactively simulate a target distribution $Q_{U,V}$ defined on a finite alphabet $\mathcal{U}\times \mathcal{V}$. To elaborate, given $d\in \mathbb{N}$, Alice and Bob wish to generate $U_d=f_d(X^d)$ and $V_d=g_d(Y^d)$, where $f_d$ and $g_d$ are potentially stochastic functions, respectively, such that $P_{U_d,V_d}$ converges to $Q_{U,V}$ as $d$ grows asymptotically with respect to an underlying distance measure, e.g., total variation distance or Kullback-Liebler divergence.

Prior works have investigated NISS scenarios in several directions, namely, they have considered \textit{decidability}, \textit{input complexity}, \textit{feasibility}, and \textit{implementability} problems. The
decidability problem focuses on the question of whether  it is possible for a Turing Machine to determine in finite time if $Q_{U,V}$ can be simulated using $(X^d,Y^d)\sim P_{X^d,Y^d}$ for some $d\in \mathbb{N}$. NISS scenarios with finite alphabet outputs  were shown to be decidable \cite{ghazi2016decidability,de2018non}. The input complexity problem in NISS focuses on quantifying the number of input samples necessary to achieve a desired total variation distance $\epsilon>0$ from a target distribution $Q_{U,V}$. 
It was shown in \cite{ghazi2016decidability} that given $P_{X,Y}$ and $Q_{U,V}$, the input complexity is $O(\exp \operatorname{poly}(\frac{1}{\epsilon}, \frac{1}{1-\rho_{X,Y}}))$, where $\rho_{X,Y}$ is the input maximal correlation.  The feasibility problem in NISS focuses on derivation of computable inner and outer bounds on the set of distributions that can be simulated in a given scenario. In this direction, a set of impossibility results for discrete-output NISS were given in \cite{kamath2016non}, where 
hypercontractivity techniques were used to provide necessary conditions for the simulatability of $Q_{U,V}$ for a given $P_{X,Y}$. These impossibility results were further improved upon in \cite{li2020boolean,shirani2023non}. The implementability problem in NISS investigates constructive mechanisms for finding the simulating functions $f_d(\cdot)$ and $g_d(\cdot)$. 
 Witsenhausen   \cite{witsenhausen1975sequences} studied  this problem in doubly-symmetric binary-output NISS (i.e., $Q_U(1)=Q_V(1)=\frac{1}{2}$), as well as scenarios where $(U,V)$ are jointly Gaussian (i.e., $Q_{U,V}$ is a Gaussian measure on $\mathbb{R}^2$), and derived an explicit algorithm to construct $f_d(\cdot)$ and $g_d(\cdot)$ with run-time $\operatorname{poly}(d)$.

The advantage of quantum protocols over their classical counterparts in various problems in computation, communications, and machine learning have been demonstrated in the literature \cite{Plenio2005,Chitambar2018,huang2021information,huang2021power,bennett2002entanglement,hsieh2010trading}. Since they can violate the Bell inequalities, entanglement, and quantum correlation enable the generation of statistical correlations that classical models cannot produce.
A question of considerable interest is to characterize conditions under which such quantum advantage manifests. For instance, it shown in \cite{bennett1999entanglement} that classical capacity of a point-to-point classical channel does not increase with shared entanglement. In contrast, for  multiple-access channels, this quantum resource  does increase the classical capacity region as was shown in \cite{leditzky2020playing,seshadri2023separation,pereg2023multiple}. This question has also been studied in the context of Shannon quantum resource theory \cite{devetak2008resource}, and quantum-quantum NISS \cite{nielsen1999conditions}.

\begin{figure*}
\begin{center}
\includegraphics[height=1.8in]{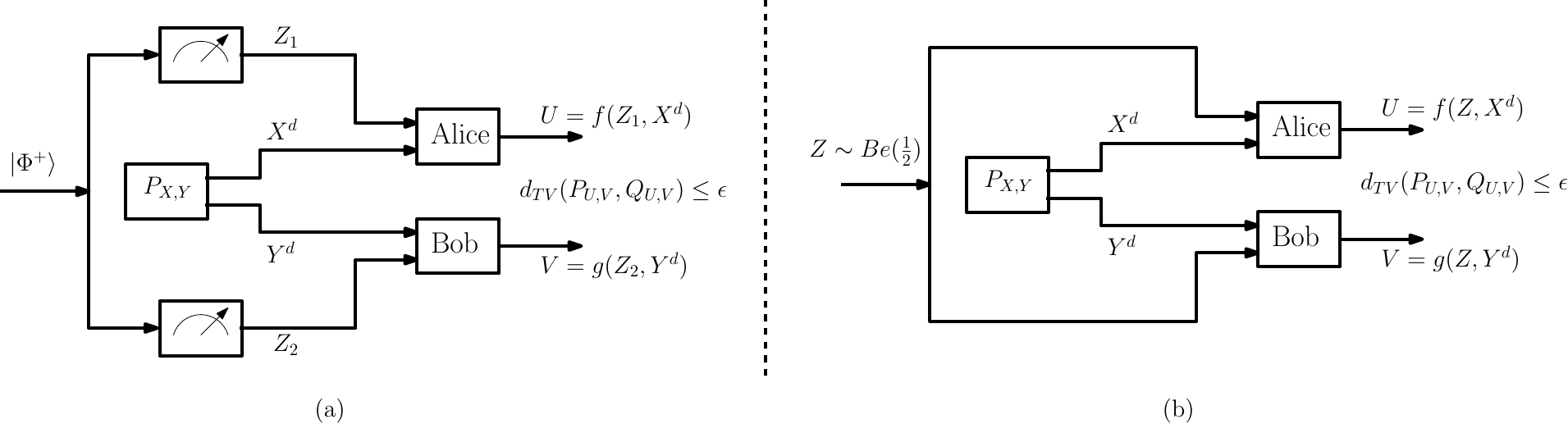}
\caption{The EA-NISS and CR-NISS scenarios: (a) In the EA-NISS scenario, Alice and Bob share a Bell state and a sequence of IID random variables $X^d$ and $Y^d$, respectively, generated according to $P_{X,Y}$; (b) In the CR-NISS scenario, Alice and Bob receive one bit of common randomness and a sequence of IID random variables $X^d$ and $Y^d$, respectively, generated according to $P_{X,Y}$.}
\label{fig:Overview}
\end{center} 
\vspace{-.3in}
\end{figure*} 

In this work, we consider an NISS scenario, where in addition to correlated classical randomness, the agents have access to a pair of fully entangled states, i.e., a shared Bell state (Figure \ref{fig:Overview}(a)). Alice and Bob each perform a measurement to produce classical outputs $Z_1$ and $Z_2$, respectively. Then, they generate $U_d=f_d(Z_1,X^d)$ and $V_d=g_d(Z_2,Y^d)$ to simulate a target distribution $Q_{U,V}$. We call this the entanglement-assisted NISS scenario (EA-NISS). We wish to characterize necessary and sufficient conditions under which there is a quantum advantage in this scenario. More precisely, we define an alternative classical NISS scenario, where, in addition to correlated sequences $(X^d,Y^d)$, Alice and  Bob have access to a common binary symmetric random variable $Z$ (Figure \ref{fig:Overview}(b)). The agents then generate $U_d=f_d(Z,X^d)$ and $V_d=g_d(Z,Y^d)$ to simulate a target distribution $Q_{U,V}$. We call this classical scenario the common randomness NISS scenario (CR-NISS). It is straightforward to see that CR-NISS is a special case of EA-NISS, where the agents choose their measurements such that they acquire a common random bit as the output. We wish to investigate whether there are conditions under which EA-NISS can simulate a larger set of distributions compared to CR-NISS. Our main contributions are summarized as follows:
\begin{itemize}[leftmargin=*]
    \item We provide a necessary and sufficient condition for simulatability in the binary-output CR-NISS, when $X^d$ and $Y^d$ are independent of each other. Furthermore, we show any distribution generated in the EA-NISS scenario is simulatable in the corresponding CR-NISS scenario, thus demonstrating no quantum advantage. (Propositions \ref{th:1} and \ref{Prop:1})
    \item We show that there is no quantum advantage in EA-NISS scenarios where $U_d$ and $V_d$ are binary, for arbitrary $P_{X,Y}$ and choice of measurements. (Proposition \ref{th:2} and Theorem \ref{th:3})
    \item We derive a necessary condition for feasibility in the non-binary-output CR-NISS when $X^d$ and $Y^d$ are independent of each other. Subsequently, we prove that the set of simulatable distributions in CR-NISS scenario forms a set of measure zero within the set of simulatable distributions in ER-NISS scenario. ( Proposition \ref{Prop:2} and Theorem \ref{th:4}) 
\end{itemize}
To summarize, the main result of this work is to show that there is indeed quantum advantage in non-binary output NISS scenarios in contrast to binary output NISS scenarios, where there is no quantum advantage.



\section{Preliminaries}
\subsection{Problem Formulation}
This work considers two NISS scenarios, a quantum-classical EA-NISS scenario (Figure \ref{fig:Overview}(a)) and purely classical CR-NISS scenario (Figure \ref{fig:Overview}(b)). Our objective is to identify scenarios under which EA-NISS setup can simulate distributions which cannot be simulated in the CR-NISS scenario.

\subsubsection{Entanglement-Assisted NISS} The scenario is shown in Figure \ref{fig:Overview}(a). Two agents, Alice and Bob, are each provided with a sequences of binary random variables $X^d$ and $Y^d$, respectively, where $d\in \mathbb{N}$. These sequences are independently and identically generated from a joint probability distribution $P_{X,Y}$. Additionally, Alice and Bob share an entangled Bell state, $|\Phi^+\rangle$. They independently perform quantum measurements, denoted as $\mathcal{M}_i, i\in \{1,2\}$, on their parts of the entangled state, yielding classical outcomes $Z_1$ and $Z_2$, respectively. The agents wish to produce variables $U_d$ and $V_d$ to simulate a target distribution $Q_{U,V}$ using (possibly stochastic) function of their respective inputs $(Z_1, X^d)$ and $(Z_2,Y^d)$.
The formal definition is given in the following.  
\begin{Definition}[\textbf{EA-NISS}] 
Consider a pair of joint distributions $P_{X,Y}$ defined on alphabet $\{-1,1\}^2$ and $Q_{U,V}$ defined on finite alphabet $\mathcal{U}\times \mathcal{V}$. Let $\mathcal{Z}$ be a finite set. 
The distribution $Q_{U,V}$ is called feasible\footnote{The terms feasible and simulatable are used interchangeably in the literature.} for $(P_{X,Y},\mathcal{Z})$ if there exist:
\\i)  sequences of measurement pairs $\mathcal{M}_{d,i} = (\Lambda_{d,i,z}, z \in \mathcal{Z})_{d\in \mathbb{N}}, i\in \{1,2\}$, where $\mathcal{Z}$ is the set of possible outputs of $\mathcal{M}_{d,i}$ and $\Lambda_{d,i,z}$ are measurement operators, and
\\ii)   sequences of (possibly stochastic) functions $f_d:\{-1,1\}^{d}\times \mathcal{Z}\to \mathcal{U}$ and $g_d:\{-1,1\}^{d}\times \mathcal{Z}\to \mathcal{V}$, such that 
\[\lim_{d\to \infty} d_{TV}(P^{(d)},Q_{U,V})=0,\] where $d_{TV}(\cdot,\cdot)$ denotes the total variation distance, 
$P^{(d)}$ is the joint distribution of $(f_d(Z_{1,d},X^d), g_d(Z_{2,d},Y^d)), d\in \mathbb{N}$, $Z_{1,d}$ and $Z_{2,d}$ are the measurement outputs acquired by Alice and Bob by applying $\mathcal{M}_{d,1}$ and $\mathcal{M}_{d,2}$ to their shared Bell state, respectively, 
and $(X^d,Y^d)$ is independent and identically distributed (IID) according to $P_{X,Y}$. The sequence $(f_d,g_d), d\in \mathbb{N}$ is called an associated sequence of functions of $Q_{U,V}$.
We denote the set of all feasible distributions for $(P_{X,Y},\mathcal{U},\mathcal{V},\mathcal{Z})$ by $\mathcal{P}_{EA}(P_{{X},{Y}},\mathcal{U},\mathcal{V},\mathcal{Z})$. The set of all feasible distribution for $P_{X,Y},\mathcal{U},\mathcal{V}$ is defined as:
\begin{align*}
    \mathcal{P}_{EA}(P_{X,Y},\mathcal{U},\mathcal{V})\triangleq \bigcup_{n\in \mathbb{N}} \mathcal{P}_{EA}(P_{X,Y},\mathcal{U},\mathcal{V},\{1,2,\cdots,n\})
\end{align*}
\label{def:1}
\end{Definition}
\subsubsection{The CR-NISS Scenario}
The scenario is shown in Figure \ref{fig:Overview}(b). In contrast with the EA-NISS scenario, here Alice and Bob are each given shared common symmetric random bit $Z$. The objective is to produce variables $U_d$ and $V_d$ simulating a target distribution $Q_{U,V}$, using (possibly stochastic) functions of their respective inputs $(Z, X^d)$ and $(Z, Y^d)$. 

\begin{Definition}[\textbf{CR-NISS}] 
Consider a joint distribution $P_{X,Y}$ defined on $\{-1,1\}^2$, and a joint distribution $Q_{U,V}$ defined on finite alphabet $\mathcal{U}\times \mathcal{V}$. The distribution $Q_{U,V}$ is called feasible for $P_{X,Y}$ if there exists a sequence of (possibly stochastic) functions $f_d:\{-1,1\}^{d+1}\to \mathcal{U}$ and $g_d:\{-1,1\}^{d+1}\to \mathcal{V}$ such that $\lim_{d\to \infty} d_{TV}(P^{(d)},Q_{U,V})=0$, where $P^{(d)}$ is the joint distribution of $(f_d(Z,X^d), g_d(Z,Y^d))$, $Z$ is a binary symmetric variable, and $(X^d,Y^d)$ is IID generated according to $P_{X,Y}$. We denote the set of all feasible distributions by $\mathcal{P}_{CR}(P_{X,Y},\mathcal{U},\mathcal{V})$.
\label{def:2}
\end{Definition}

Note that the EA-NISS scenario may be viewed as a generalization of the CR-NISS scenario. To elaborate, let us take the pair of measurements performed by Alice and Bob:
\begin{align*}
    \Lambda_{d,i,0}= 
    \begin{bmatrix}
        1&0\\
        0&0
    \end{bmatrix},
    \quad 
        \Lambda_{d,i,1}= 
    \begin{bmatrix}
        0&0\\
        0&1
    \end{bmatrix},
    \quad i\in \{1,2\}.
\end{align*}
Then, Alice and Bob observe the classical measurement output $Z=Z_1=Z_2$, where $Z$ is a binary symmetric variable. Hence, this choice of measurement recovers the CR-NISS problem. Consequently, $\mathcal{P}_{CR}(P_{X,Y},\mathcal{U},\mathcal{V})\subseteq \mathcal{P}_{EA}(P_{X,Y},\mathcal{U},\mathcal{V})$. We wish to investigate whether there are scenarios in which $\mathcal{P}_{EA}(P_{X,Y},$ $\mathcal{U},\mathcal{V})$ strictly contains $\mathcal{P}_{CR}(P_{X,Y},\mathcal{U},\mathcal{V})$.


\subsection{Boolean Fourier Expansion}
\label{sec:BFE}
The analysis provided in the next sections relies on the Boolean Fourier expansion techniques which are briefly described in the following. A more complete discussion on Boolean Fourier analysis can be found in \cite{ o2014analysis,Wolf2008}.

Consider a vector $X^d$ consisting of IID variables with alphabet $\mathcal{X}=\{-1,1\}$, where $P_X(1)=p\in (0,1)$. Let $\mu_X=2p-1$ and $\sigma_X= 2\sqrt{p(1-p)}$ be the mean and standard deviation, respectively. 
Let $\mathcal{L}_{X^d}$ be the space of functions $f_d:\{-1,1\}^d \to \mathbb{R}$ equipped with the inner-product \[\langle f_d(\cdot), g_d(\cdot)\rangle \triangleq \mathbb{E}(f_d(X^d)g_d(X^d)),\quad  f_d(\cdot),g_d(\cdot)\in \mathcal{L}_{X^d}.\] 
Then, the following collection of \textit{parity} functions forms an orthonormal basis for $\mathcal{L}_{X^d}$;
\begin{align*}
\phi_{\mathcal{S}}(x^d)\triangleq \prod_{i\in \mathcal{S}}\frac{x_i-\mu_X}{\sigma_X},\qquad  x^d \in \{-1,1\}^d, \mathcal{S}\subseteq [d].
\end{align*}
For an arbitrary $f_d\in \mathcal{L}_{X^d}$, the Boolean Fourier expansion is:
\begin{align*}
 f_d(x^d)=\sum_{\mathcal{S}\subseteq [d]} f_\mathcal{S} \phi_\mathcal{S}(x^d), \quad \text{for all} \ x^d \in \{-1,1\}^d,
\end{align*}
where $f_\mathcal{S}\!=\!\langle f_d(.), \phi_\mathcal{S}(.) \rangle, \mathcal{S}\subseteq [d]$ are the Fourier coefficients.

Consider a pair of correlated binary random variables $(X^d,Y^d)$ distributed according to the joint distribution $P_{X,Y}$, and a pair of functions $f_d,g_d\in \mathcal{L}_{X^d}\times \mathcal{L}_{Y^d}$. We have:
  \begin{align}\label{eq:1}
  & \mathbb{E}(f_d(X^d)g_d(Y^d))=\sum_{\mathcal{S}\subseteq[d]} f_{\mathcal{S}}g_{\mathcal{S}}\mathbb{E}\big(\prod_{i\in \mathcal{S}} \frac{(X_i-\mu_X)(Y_i-\mu_Y)}{\sigma_X\sigma_Y}\big)
     = \sum_{\mathcal{S}\subseteq [d]} f_{\mathcal{S}}g_{\mathcal{S}} \rho^{|\mathcal{S}|},
  \end{align}
where $\rho\triangleq  \mathbb{E}\left(\frac{(X-\mu_X)(Y-\mu_Y)}{\sigma_X\sigma_Y}\right)$ is the Pearson correlation coefficient between $X$ and $Y$. 

\section{Binary-Output NISS Scenarios}
In this section, we restrict our analysis to binary-output NISS scenarios, where $\mathcal{U}=\mathcal{V}=\{-1,1\}$. We show that any distribution that can be simulated in a binary-output EA-NISS scenario, can also be simulated in the corresponding CR-NISS scenario. Thus, there is no quantum advantage in binary-output NISS scenarios. The proof is provided in several steps. We start with a special case, where $(X^d,Y^d)$ are classical local randomness, i.e., $X$ and $Y$ are independent of each other. We incrementally build upon the ideas introduced in analyzing this special case to prove the result for general binary-output NISS. 
\subsection{Binary Measurements and Local Classical Randomness}\label{sec: Binary_output_classical}
As a first step, we consider the scenario where the measurements $\mathcal{M}_{d,i},i \in \{1,2\}, d\in \mathbb{N}$ performed by Alice and Bob have binary-valued outputs, i.e., $\mathcal{Z}=\{-1,1\}$. We further assume that, in addition to the shared Bell state in the EA-NISS scenario and the classical common random bit in the CR-NISS scenario, 
Alice and Bob only have access to classical local randomness, i.e., $P_{X,Y}=P_XP_Y$. We show that in this scenario there is no quantum advantage. That is, we show that $\mathcal{P}_{EA}(P_XP_Y,\mathcal{U},\mathcal{V}, \{-1,1\})=\mathcal{P}_{CR}(P_XP_Y,\mathcal{U},\mathcal{V})$ for all $P_X,P_Y$ defined on binary alphabets. The following proposition formalizes the main claim in this section.
\begin{Proposition}\label{th:1}
 Let $P_X$ and $P_Y$ be two probability distributions on $\{-1,1\}$, and let $\mathcal{U}=\mathcal{V}=\mathcal{Z}=\{-1,1\}$. Then, $\mathcal{P}_{EA}(P_XP_Y,\mathcal{U},\mathcal{V}, \mathcal{Z})=\mathcal{P}_{CR}(P_XP_Y,\mathcal{U},\mathcal{V})$
\end{Proposition}

To prove Proposition \ref{th:1}, we first characterize the set of distributions which can be simulated classically. The following proposition provides a necessary and sufficient condition for such distributions.

\begin{Proposition}\label{Prop:1}
Let $P_X$ and $P_Y$ be two probability distributions defined on $\{-1,1\}$ such that $P_X(1),P_Y(1)\notin\{0,1\}$ and let $\mathcal{U}=\mathcal{V}=\{-1,1\}$. Then, 
    \begin{align} \label{eq: classical_bound_lemma}
   &     \mathcal{P}_{CR}(P_XP_Y,\mathcal{U},\mathcal{V}) = \bigcup_{a,b \in [0,1]}\Big\{ Q_{U,V} \Big| Q_U(1)=a, Q_V(1)=b, 
         |Q_{U,V}(-1,-1)+Q_{U,V}(1,1)- \zeta_{a,b}| \leq  2\beta_{ab}\Big\},
    \end{align}
    where
    \begin{align*}
    &\zeta_{a,b}\triangleq 2ab-a-b+1, \quad a,b\in [0,1],
    \\& \beta_{a,b} \triangleq  \min\{a, (1 - a)\} \min\{b, (1 - b)\},  \quad a,b\in [0,1].
    \end{align*}
\end{Proposition}
The proof is based on Fourier expansion techniques (Section \ref{sec:BFE}) and is given in Appendix \ref{App:prop:1}. To prove Proposition \ref{th:1}, it suffices to show that any feasible distribution for the EA-NISS scenario satisfies Equation \eqref{eq: classical_bound_lemma}. This is proved in Appendix \ref{App:th:1}. The next corollary follows from the fact that any given  $P'_XP'_Y$ may be transformed to a desired $P_XP_Y$ locally (e.g. \cite{VonNeumann1951}). 
 \begin{Corollary}
 \label{cor:1}
 Let $P_X,P'_X,P_Y$ and $P'_Y$ be probability distributions on $\{-1,1\}$, and let $\mathcal{U}=\mathcal{V}=\mathcal{Z}=\{-1,1\}$. Then, $\mathcal{P}_{EA}(P_XP_Y,\mathcal{U},\mathcal{V}, \mathcal{Z})=\mathcal{P}_{CR}(P'_XP'_Y,\mathcal{U},\mathcal{V})$.\footnote{We are grateful to Professor Masahito Hayashi for the engaging discussions following the initial publication of this result on the arXiv repository. Professor Hayashi provided an alternative proof specific to the binary-output measurement case that avoids the use of Fourier transforms and does not require the condition $\mathcal{U} = \mathcal{V}=\{-1,1\}$.}     
 \end{Corollary}

\subsection{Binary Measurements and Correlated Randomness}
We consider binary-otuput measurements $\mathcal{M}_{d,i},i\in \{1,2\},d\in \mathbb{N}$, and arbitrary $P_{X,Y}$. The following theorem states the main result of the section. 
\begin{Proposition}\label{th:2}
 Let $P_{X,Y}$ be a joint probability distribution on $\{-1,1\}^2$, and let $\mathcal{U}=\mathcal{V}=\mathcal{Z}=\{-1,1\}$. Then, $\mathcal{P}_{EA}(P_{X,Y}, \mathcal{U},\mathcal{V},\mathcal{Z})=\mathcal{P}_{CR}(P_{X,Y},\mathcal{U},\mathcal{V})$.
\end{Proposition}

The proof uses the fact that in the EA-NISS scenario, for any given realization $(X^d,Y^d)=(x^d,y^d)$, the output of Alice $U_{x^d}$ and Bob $V_{y^d}$ are (possibly stochastic) functions of only their measurement outputs $Z_1$ and $Z_2$, respectively. So, for each $(X^d,Y^d)=(x^d,y^d)$, using Proposition \ref{th:1}, the outputs can be simulated using classical processing of a common random bit. The proof follows by \textit{patching together} each of these classical functions for different realizations $(X^d,Y^d)=(x^d,y^d)$ to construct a classical simulation protocol. The complete proof is given in Appendix \ref{App:th:2}. It can be noted that in the proof of Proposition \ref{th:2}, we have not used the fact that the measurement operators are chosen by Alice and Bob prior to the observation of $X^d$ and $Y^d$. In fact, the following corollary, which follows from the proof of Proposition \ref{th:2} states that there is no quantum advantage in this NISS scenario, even if the agents make their choice of measurement dependent on their observed classical sequences. 
\begin{Corollary}
    Let $\mathcal{Z}=\{-1,1\}$, and consider $P_{X,Y}$ and $Q_{U,V}$ defined on $\{-1,1\}^2$, such that there exist
    sequences of collection of measurement pairs
    \begin{align*}
        &\mathcal{M}_{d,x^d,1} = (\Lambda_{d,x^d,1,z}, z \in \mathcal{Z})_{d\in \mathbb{N}, x^d\in \{-1,1\}^d}
        \\& \mathcal{M}_{d,y^d,2} = (\Lambda_{d,y^d,2,z}, z \in \mathcal{Z})_{d\in \mathbb{N}, y^d\in \{-1,1\}^d},
    \end{align*} such that $Q_{U,V}$ can be simulated by Alice and Bob using source sequences $(Z_{1,d},X^d)$ and $(Z_{2,d},Y^d)$, respectively, where  $Z_{1,d}$ and $Z_{2,d}$ are the measurement outputs acquired using $\mathcal{M}_{d,X^d,1}$ and $\mathcal{M}_{d,Y^d,2}$. Then,
$Q_{U,V}\in \mathcal{P}_{CR}(P_{X,Y},\mathcal{U},\mathcal{V})$.
\end{Corollary}

\subsection{Non-Binary Measurements}
In this section, we extend the results of previous sections to non-binary valued measurements. The following states the main result of this section. 
\begin{Theorem}\label{th:3}
 Let $P_{X,Y}$ be a joint probability distribution on $\{-1,1\}^2$ and $\mathcal{U}=\mathcal{V}=\{-1,1\}$. Then, $\mathcal{P}_{EA}(P_{X,Y},\mathcal{U},\mathcal{V})=\mathcal{P}_{CR}(P_{X,Y},\mathcal{U},\mathcal{V})$.
\end{Theorem}
The proof follows by similar arguments as in the proofs of Propositions \ref{th:1} and \ref{th:2}. We provide an outline in the following. First, we consider the case where the agents do not have access to correlated randomness, i.e., $P_{X,Y}=P_XP_Y$. We use the well-known fact that any concatenation of non-binary quantum measurement followed by local classical processing of the measurement output into binary outputs can be mapped to a one shot binary-output quantum measurement. Thus $\mathcal{P}_{EA}(P_{X}P_Y,\mathcal{U},\mathcal{V})=\mathcal{P}_{EA}(P_{X}P_{Y},\mathcal{U},\mathcal{V},\{-1,1\})$.
We conclude from Proposition \ref{th:1} that $\mathcal{P}_{EA}(P_{X}P_{Y},\mathcal{U},\mathcal{V})=\mathcal{P}_{CR}(P_{X}P_{Y},\mathcal{U},\mathcal{V})$ for all marginal distributions $P_X,P_Y$. Next, for the general scenario, where $X$ and $Y$ are not independent, we follow the steps in the proof of Proposition \ref{th:2}. That is, from the previous arguments, we conclude that for any realization $x^d,y^d$ of the classical inputs, the output in the EA-NISS scenario can be simulated using a classical random bit. So, for any fixed $x^d,y^d\in \{-1,1\}^d$, we can construct classical functions $f_{x^d,y^d}(Z)$ and $g_{x^d,y^d}(Z)$ for the CR-NISS scenario which simulate the output of Alice and Bob in the EA-NISS scenario. Next, we define the collection of functions $f^+_{x^d,y^d}, g^+_{x^d,y^d}, g^-_{x^d,y^d}$, and $p_{ts}$ as in the proof of Proposition \ref{th:2} based on $f_{x^d,y^d}(Z,X^d)$ and $g_{x^d,y^d}(Z,Y^d)$. 
We define the CR-NISS simulation scheme as follows. Alice observes $X^d=x^d$, and computes $f^+_{x^d}(Z)$. She uses independent copies of $X_i, i>d$ to generate U such that $P_{U}(1)= \frac{1+f^+_{x^d}(Z)}{2}$ and
$P_{U}(-1)= \frac{1-f^+_{x^d}(Z)}{2}$. Bob first generates a binary random variable $T$ such that $P_T(1)= p_{ts}$ and $P_T(-1)=1-p_{ts}$, using copies of $Y_i, i>d$ which are non-overlapping with those used by Alice.  Bob then observes $Y^d=y^d$, and uses additional non-overlapping copies of $Y_i,i>d$ to generate a binary random variable $V^+$ such that $P_{V^+}(1)= \frac{1+g^+_{y^d}(Z)}{2}$ and
$P_{V^+}(-1)= \frac{1-g^+_{y^d}(Z)}{2}$, and a binary random variable $V^-$ such that  $P_{V^-}(1)= \frac{1+g^-_{y^d}(Z)}{2}$ and
$P_{V^-}(-1)= \frac{1-g^-_{y^d}(Z)}{2}$. If $T=1$, Bob outputs $V=V^+$ and if $T=-1$, Bob outputs $V=V^-$. As shown in the proof of Proposition \ref{th:2}, this scheme guarantees that the distribution generated by Alice and Bob in the CR-NISS scenario simulates the one generated in the EA-NISS scenario. 

\section{Non-Binary Output NISS Scenarios}
In the previous sections, we have shown that there is no quantum advantage in NISS scenarios simulating binary-output variables. In this section, we demonstrate quantum advantage in general EA-NISS scenarios over CR-NISS scenarios. We further show that the set of distributions generated in CR-NISS has measure zero among the set of distributions generated in EA-NISS.
We begin with a simple example illustrating a specific joint distribution that can be simulated within the EA-NISS framework but is not simulatable under the CR-NISS scenario.  
\begin{Example}
    Consider a NISS scenario in which agents Alice and Bob aim to simulate an output distribution $P_{U,V}$ on finite sets $\mathcal{U}, \mathcal{V} = \{1, 2, 3\}$. Further assume that other than the shared binary symmetric variable $Z$ in the CR-NISS scenario and the shared Bell pair in the EA-NISS scenario,
    they only have access to local randomness, i.e.,  $P_{X,Y} = P_X P_Y$.
    
    \textbf{EA-NISS Scenario:} The measurement operators used by Alice and Bob are defined as follows:
    \begin{align*}
    & {\Lambda}_{1,1}={\Lambda}_{2,1}=\frac{2}{3}
    \begin{bmatrix}
        & 0 & 0 \\
        & 0 & 1
    \end{bmatrix}, \qquad {\Lambda}_{1,2}={\Lambda}_{2,2}=\frac{2}{3}
    \begin{bmatrix}
        & \frac{3}{4} & \frac{\sqrt{3}}{4} \\
        & \frac{\sqrt{3}}{4} & \frac{1}{4}
    \end{bmatrix}, \qquad {\Lambda}_{1,3}={\Lambda}_{2,3}=\frac{2}{3}
    \begin{bmatrix}
        & \frac{3}{4} & -\frac{\sqrt{3}}{4} \\
        & -\frac{\sqrt{3}}{4} & \frac{1}{4}
    \end{bmatrix}. 
\end{align*}
The joint distribution $P_{Z_1,Z_2}$ resulting from their measurement outputs is computed as:
    \begin{align*}
        \mathbb{P}_{Z_1,Z_2}(z_1, z_2) = \langle \Phi^+ | \Lambda_{1,z_1} \otimes \Lambda_{2,z_2} | \Phi^+ \rangle = \frac{1}{2} \text{Vec}(\Lambda_{1,z_1})^\top \text{Vec}(\Lambda_{2,z_2}).
    \end{align*}
    Applying the equation, we get:
\begin{align*}
 &   P_{Z_1,Z_2}(1,1)= \frac{2}{9} ,\quad  P_{Z_1,Z_2}(1,2)= \frac{1}{18} , \quad  P_{Z_1,Z_2}(1,3)= \frac{1}{18}  \\
    &P_{Z_1,Z_2}(2,1)= \frac{1}{18} ,\quad  P_{Z_1,Z_2}(2,2)= \frac{2}{9} , \quad  P_{Z_1,Z_2}(2,3)= \frac{1}{18}  \\
    &P_{Z_1,Z_2}(3,1)= \frac{1}{18} ,\quad  P_{Z_1,Z_2}(3,2)= \frac{1}{18} , \quad  P_{Z_1,Z_2}(3,3)= \frac{2}{9}.
\end{align*}

The agents output $U=Z_1$ and $V=Z_2$. We note the joint distribution matrix of $(U,V)$ has rank three. 

\textbf{CR-NISS Scenario:} We argue that the distribution generated in the above EA-NISS setup cannot be generated in CR-NISS. To see this, note that since Alice and Bob only have access to local randomness, their outputs $U$ and $V$ follow the distribution $P_{UV}(\cdot,\cdot)=\sum_{z\in \{0,1\}}P_Z(z) P_{U|Z}(\cdot|z)P_{V|Z}(\cdot|z)$, where $P_Z$ is $Be(\frac{1}{2})$. Thus, the resulting joint probability distribution of $U$ and $V$ has rank at most equal to two. This complete the proof of quantum advantage in this particular NISS scenario.  
\end{Example}
The following theorem formally shows that the set of distributions simulatable in CR-NISS has measure zero in those generated by EA-NISS. 

\begin{Theorem}
\label{th:4}
    Let $P_{X,Y}$ be a joint distribution defined on $\{-1,1\}^2$ and let $\mathcal{U}, \mathcal{V}$ be non-binary finite alphabets. The set $\mathcal{P}_{CR}(P_{X,Y}, \mathcal{U}, \mathcal{V})$ forms a set of measure zero within the set  $\mathcal{P}_{EA}(P_{X,Y}, \mathcal{U}, \mathcal{V})$.
\end{Theorem}

The proof can potentially be provided directly using the rank of the joint distribution matrix as in the previous example. Alternatively, the Fourier transform machinery developed in the previous sections to derive a necessary condition on the set of distributions generated in CR-NISS as in the following proposition. The proof of the theorem follows in a straightforward manner.
\begin{Proposition}
\label{Prop:2}
    Let $P_{X},P_Y$ be two probability distributions on $\{-1,1\}^2$ and let $\mathcal{U}$ and $\mathcal{V}$ be finite sets. Then, for any $Q_{U,V}\in \mathcal{P}_{CR}(P_{X}P_Y,\mathcal{U},\mathcal{V})$, the following holds for all $(i,j)\in \mathcal{U}\times \mathcal{V}$:
        \begin{align} 
&        (Q_{UV}(i,i)-Q_U(i)Q_V(i))(Q_{UV}(j,j)-Q_U(j)Q_V(j)=\label{eq:constraint_ternary}
        (Q_{UV}(i,j)\!-\!Q_U(i)Q_V(j))(Q_{UV}(j,i)-Q_U(j)Q_V(i)).
    \end{align}
\end{Proposition}
\begin{proof}
    Let $(f_d,g_d)_{d\in \mathbb{N}}$ be the associated functions of $Q_{U,V}$, and fix $d\in \mathbb{N}$. Let us define 
    \begin{align*}
        &f_{d,u}(Z,X^d)\triangleq \mathbbm{1}(f_{d}(Z,X^d)=u), \quad u\in \mathcal{U}\\
        &g_{d,v}(Z,Y^d)\triangleq \mathbbm{1}(g_{d}(Z,Y^d)=v),\quad  v\in \mathcal{V}.
    \end{align*}
    Since $f_{d,u},g_{d,v}$ are Boolean functions, as explained in the proof of Proposition \ref{th:1}, we can use Boolean Fourier expansion to decompose them into parity functions:
    \begin{align*}
        &f_{d,u}(Z,X^d)=f_{d,u,0,\phi}+f_{d,u,1,\phi}Z +\sum_{\mathcal{S}\subseteq [d], \mathcal{S}\neq \phi} (f_{d,u,0,\mathcal{S}}+Zf_{d,u,1,\mathcal{S}})\prod_{i\in \mathcal{S}}\frac{X_i-\mu_X}{\sigma_X}, \quad u \in \mathcal{U},
        \\& g_{d,v}(Z,Y^d)=g_{d,v,0,\phi}+g_{d,v,1,\phi}Z+\sum_{\mathcal{S}\subseteq [d],\mathcal{S}\neq \phi} (g_{d,v,0,\mathcal{S}}+Zg_{d,v,1,\mathcal{S}})\prod_{i\in \mathcal{S}}\frac{Y_i-\mu_Y}{\sigma_Y},\quad v \in \mathcal{V}.
    \end{align*}
    Consequently, using the independence of $X_i$ and $Y_i$: 
    \begin{align*}
        \mathbb{E}(f_{d,u}(Z,X^d)g_{d,v}(Z,Y^d))=  f_{d,u,\phi}g_{d,v,\phi}+f_{d,u,1,\phi}g_{d,v,1,\phi}. 
    \end{align*}
    On the other hand:
    \begin{align*}
        &f_{d,u,\phi}= \mathbb{E}(\mathbbm{1}(U=u))= 2Q_U(u)-1,\\
        &g_{d,v,\phi}= \mathbb{E}(\mathbbm{1}(V=v))= 2Q_V(v)-1,\\
        & \mathbb{E}(f_{d,u}(Z,X^d)g_{d,v}(Z,Y^d))= \mathbb{E}(\mathbbm{1}(U=u,V=v))
      = 4Q_{U,V}(u,v)-2Q_U(u)-2Q_V(v)+1. 
    \end{align*}
    As a result, 
    \begin{align*}
        f_{d,u,1,\phi}g_{d,v,1,\phi}= 4(Q_{U,V}(u,v)-Q_U(u)Q_V(v)). 
    \end{align*}
    For any given $(i,j)\in \mathcal{U}\times \mathcal{V}$, we have:
    \begin{align*}
        &f_{d,i,1,\phi}g_{d,i,1,\phi}= 4(Q_{U,V}(i,i)-Q_U(i)Q_V(i))\\
        &f_{d,j,1,\phi}g_{d,j,1,\phi}= 4(Q_{U,V}(j,j)-Q_U(j)Q_V(j))\\
        &\Rightarrow f_{d,i,1,\phi}g_{d,i,1,\phi}f_{d,j,1,\phi}g_{d,j,1,\phi}=
     16 (Q_{U,V}(i,i)-Q_U(i)Q_V(i))(Q_{U,V}(j,j)-Q_U(j)Q_V(j)).
    \end{align*}
    Similarly,
        \begin{align*}
        &f_{d,i,1,\phi}g_{d,j,1,\phi}= 4(Q_{U,V}(i,j)-Q_U(i)Q_V(j))\\
        &f_{d,j,1,\phi}g_{d,i,1,\phi}= 4(Q_{U,V}(j,i)-Q_U(j)Q_V(i))\\
        &\Rightarrow f_{d,i,1,\phi}g_{d,i,1,\phi}f_{d,j,1,\phi}g_{d,j,1,\phi}=
         16 (Q_{U,V}(i,j)-Q_U(i)Q_V(j))(Q_{U,V}(j,i)-Q_U(j)Q_V(i)).
    \end{align*}
    Combining the two results, we get:
    \begin{align*}
         (Q_{U,V}(i,i)-Q_U(i)Q_V(i))(Q_{U,V}(j,j)-Q_U(j)Q_V(j)=(Q_{U,V}(i,j)-Q_U(i)Q_V(j))(Q_{U,V}(j,i)-Q_U(j)Q_V(i)).
    \end{align*}
    This completes the proof. 
\end{proof}

\section{Conclusion}
Two variations of the standard NISS scenario, namely EA-NISS and CR-NISS, were considered.  It was shown that for binary-output NISS scenarios, the set of feasible distributions for EA-NISS and CR-NISS are equal with each other. Hence, there is no quantum advantage in these EA-NISS scenarios. To this end, first a
a necessary and sufficient condition for feasibility in the binary-output CR-NISS was provided.
   It was shown that any distribution generated in an EA-NISS scenario satisfies this condition  and is hence feasible for the corresponding CR-NISS. 
For non-binary output NISS scenarios, it was shown that the set of distributions that are feasible in CR-NISS forms a set of measure zero within the set of feasible distribution in EA-NISS scenario.
\newpage
\IEEEtriggeratref{41}
\bibliographystyle{IEEEtran}
\bibliography{References_Quantum}
\newpage

\begin{appendices}
\section{Proof of Proposition \ref{th:1}}
\label{App:th:1}
Let us fix a blocklength $d\in \mathbb{N}$, simulating functions $f_d,g_d$, and positive operator valued measuremtns (POVM) $\mathcal{M}_{d,i}=\{\Lambda_{d,i,z}, z\in \{-1,1\}\}, i\in \{1,2\}$, where for each $i\in \{1,2,\}$, $\Lambda_{d,i,z}\in \mathbb{R}^{2\times 2},z\in \{-1,1\}$ are self-adjoint positive semi-definite measurement operators satisfying the completeness relation, i.e., $0\leq \Lambda_{d,i,z}\leq \mathbb{I}_{2\times 2}$, and $\Lambda_{d,i,1}+\Lambda_{d,i,2}=\mathbb{I}_{2\times 2}$, where is the $2\times 2$ identity matrix. Let $Z_1,Z_2$ denote the classical measurement outcomes of Alice and Bob, respectively.
Note that to prove the theorem, it suffices to show that $P_{Z_1,Z_2}\in \mathcal{P}_{CR-NISS}(P_XP_Y,\mathcal{U},\mathcal{V})$, since any post-processing of $Z_1,Z_2$ in the classical domain in the EA-NISS setting can be performed in the CR-NISS setting as well, hence any feasible distribution in EA-NISS is feasible in CR-NISS as long as the underlying $P_{Z_1,Z_2}$ acquired in the EA-NISS process is feasible in the CR-NISS setting. Consequently, it suffices to prove $P_{Z_1,Z_2}$ satisfies Equation \eqref{eq: classical_bound_lemma}. We provide the proof for $P_{Z_1}(1),P_{Z_2}(1)\leq \frac{1}{2}$. The proof can then be extended to general $P_{Z_1,Z_2}$ by flipping the order of the measurement operators if necessary to ensure $P_{Z_1}(1),P_{Z_2}(1)\leq \frac{1}{2}$ and applying the original proof arguments.
To prove the result the case when  $P_{Z_1}(1),P_{Z_2}(1)\leq \frac{1}{2}$, not that we have:
\begin{align*}
    \mathbb{P}_{Z_1,Z_2}(z_1,z_2)=\langle\Phi^+ |  \Lambda_{d,1,z_1} \otimes  \Lambda_{d,2,z_2} | \Phi^+ \rangle =  \frac{1}{2}Vec( \Lambda_{d,1,z_1})^T  Vec(\Lambda_{d,2,z_2}),  
\end{align*}
where $Vec(
\begin{bmatrix}
    a&b\\
    c&d
\end{bmatrix}= [a,b,c,d]
), a,b,c,d\in \mathbb{C}$. Consequently,
\begin{align*}
&    a\triangleq P_{Z_1}(1)= \frac{1}{2}Vec( \Lambda_{d,1,1})^TVec(\Lambda_{d,2,1})+ \frac{1}{2}Vec( \Lambda_{d,1,1})^T Vec(\Lambda_{d,2,-1})
\\&= \frac{1}{2}Vec( \Lambda_{d,1,1})^T Vec(\Lambda_{d,2,-1}+\Lambda_{d,2,1})=  \frac{1}{2}Vec( \Lambda_{d,1,1})^T Vec(\mathbb{I}_{2\times 2}),
\end{align*}
where we have used the completeness relation to conclude that $\Lambda_{d,2,-1}+\Lambda_{d,2,1}=\mathbb{I}_{2\times 2}$. Similarly, 
\begin{align*}
    b\triangleq P_{Z_2}(1)= \frac{1}{2}Vec( \mathbb{I}_{2\times 2})^T Vec(\Lambda_{d,2,1}).
\end{align*}
On the other hand:
\begin{align}
   & P(Z_1=Z_2)= P_{Z_1,Z_2}(-1,-1)+P_{Z_1,Z_2}(1,1)= 1-P_{Z_1}(1)-P_{Z_2}(1)+ 2P_{Z1,Z_2}(1,1)\nonumber
    \\&= 1-a-b+Vec(\Lambda_{d,1,1})^T Vec(\Lambda_{d,2,1}).\label{eq:th1:1}
\end{align}
Furthermore, since $a,b\leq \frac{1}{2}$ by assumption, we have $\beta_{a,b}= ab$. So, it suffices to show that:
\begin{align*}
    | P(Z_1=Z_2)- \zeta_{a,b}|\leq 2ab.
\end{align*}
Replacing $P(Z_1=Z_2)$ using Equation \eqref{eq:th1:1}, it suffices to prove:
\begin{align*}
    |Vec(\Lambda_{d,1,1})^T Vec(\Lambda_{d,2,1})-2ab|\leq 2ab.
\end{align*}
Let us take
\begin{align*}
    \Lambda_{d,1,1}=
    \begin{bmatrix}
        e&f\\f^*&g
    \end{bmatrix},\quad \Lambda_{d,2,1}= \begin{bmatrix}
        h&n\\n^*&q
    \end{bmatrix},
\end{align*}
where $e,g,h,q\in \mathbb{R}$ and $n,f\in \mathbb{C}$. We have:
\begin{align*}
   & ab= \frac{1}{4}Vec(\Lambda_{d,1,1})^T Vec( \mathbb{I}_{2\times 2})\times Vec( \mathbb{I}_{2\times 2})^T Vec(\Lambda_{d,2,1})
    =\frac{(e+g)(h+q)}{4}
    \\& Vec(\Lambda_{d,1,1})^T Vec(\Lambda_{d,2,1})= eh+fn+f^*n^*+gq.
\end{align*}
So, it suffices to show that:
\begin{align*}
   0\leq  eh+fn+f^*n^*+gq\leq (e+g)(h+q),
\end{align*}
equivalently:
\begin{align*}
&    fn+f^*n^*\leq eq+gh,
\qquad  0\leq eh+fn+f^*n^*+gq.
\end{align*}
The positive-semidefinite property of $\Lambda_{d,1,1}$ and $\Lambda_{d,2,1}$ implies that $|f|^2\leq eg$ and $|n|^2\leq qh$, consequently, $fn+f^*n^*\leq 2|f||n|\leq 2\sqrt{egqh}$. So, it suffices to show that:
\begin{align*}
&    2\sqrt{egqh}\leq eq+gh,
\qquad 2\sqrt{egqh}\leq eh+gq
\end{align*}
which is straightforward to verify by noting that $eq+gh$ and $eh+gq$ are non-negative and taking the square of both sides.
\qed
\section{Proof of Proposition \ref{Prop:1}}
\label{App:prop:1}
Let us define
\begin{align*}
    &\mathcal{P}'=\bigcup_{a,b \in [0,1]}\Big\{ Q_{U,V} \mid Q_U(1)=a, Q_V(1)=b, \nonumber \\
        & |Q_{U,V}(-1,-1)+Q_{U,V}(1,1)- \zeta_{a,b}| \leq  2\beta_{ab}\Big\}.
\end{align*}
We show that $\mathcal{P}_{CR}(P_XP_Y,\mathcal{U},\mathcal{V})=\mathcal{P}'$ by showing that each set is a subset of the other.  \\\textbf{Proof of $\mathcal{P}_{CR}(P_XP_Y,\mathcal{U},\mathcal{V})\subseteq \mathcal{P}'$}:
\\ Let $Q_{U,V}\in \mathcal{P}_{CR}(P_XP_Y,\mathcal{U},\mathcal{V})$ be a feasible distribution for $P_XP_Y$, and let $(f_d,g_d)_{d\in \mathbb{N}}$ be its associated sequence of functions. Let us fix $d\in \mathbb{N}$ and denote $U_d= f(Z,X^d)$ and $V_d=g(Z,Y^d)$. As discussed in Section \ref{sec:BFE}, the Fourier parity functions for $\mathcal{L}_{Z,X^d}$ are given by
\begin{align*}
  &  \phi_{1,\mathcal{S}}(Z,X^d)\triangleq Z\prod_{i\in \mathcal{S}}\frac{X_i-\mu_X}{\sigma_X},\quad
  \\&\phi_{0,\mathcal{S}}(Z,X^d)\triangleq \prod_{i\in \mathcal{S}}\frac{X_i-\mu_X}{\sigma_X},
\end{align*}
where $\mathcal{S}\subseteq [d]$. The parity functions form an orthonormal basis for $\mathcal{L}_{Z,X^d}$. Similarly, the Fourier parity functions for $\mathcal{L}_{Z,Y^d}$ are given by
\begin{align*}
  &  \psi_{1,\mathcal{S}}(Z,Y^d)\triangleq Z\prod_{i\in \mathcal{S}}\frac{Y_i-\mu_Y}{\sigma_Y},
  \\& \psi_{0,\mathcal{S}}(Z,Y^d)\triangleq \prod_{i\in \mathcal{S}}\frac{Y_i-\mu_Y}{\sigma_Y},
\end{align*}
Consequently, we use the Fourier expansion to write:
\begin{align*}
    &U_d= f_{0,\phi}+  f_{1,\phi}Z+
\sum_{\mathcal{S}\subseteq [d], \mathcal{{S}\neq \phi}} (f_{0,\mathcal{S}}+f_{1,\mathcal{S}} Z)\left(\prod_{i\in \mathcal{S}}\frac{X_i-\mu_X}{\sigma_X}\right),\\
     &V_d= g_{0,\phi}+g_{1,\phi}Z+
\sum_{\mathcal{S}\subseteq [d],\mathcal{S}\neq\phi} (g_{0,\mathcal{S}}+g_{1,\mathcal{S}} Z)\left(\prod_{i\in \mathcal{S}}\frac{Y_i-\mu_Y}{\sigma_Y}\right),\\
\end{align*}
where $f_{0,\mathcal{S}},f_{1,\mathcal{S}}, \mathcal{S}\subseteq [d]$ and $g_{0,\mathcal{S}},g_{1,\mathcal{S}}, \mathcal{S}\subseteq [d]$ are the Fourier coefficients of the functions $f(Z,X^d)$ and $g(Z,Y^d)$, respectively. It can be noted that by linearity of expectation:
\begin{align}
    \mathbb{E}(U_d)= f_{0,\phi}, \quad \mathbb{E}(V_d)= g_{0,\phi}.
    \label{eq:Prop1:1}
\end{align}
Furthermore,
\begin{align}
\nonumber   & \mathbb{E}(U_dV_d)= f_{0,\phi}g_{0,\phi}+f_{1,\phi}g_{1,\phi}\mathbb{E}(Z^2) +\sum_{\substack{\mathcal{S},\mathcal{S}'\subseteq [d]\\ \mathcal{S},\mathcal{S}'\neq \phi}}\Bigg(
\nonumber     f_{0,\mathcal{S}}g_{0,\mathcal{S}}\prod_{i\in \mathcal{S}}\Big(\frac{\mathbb{E}(X_i)-\mu_X}{\sigma_X}\frac{\mathbb{E}(Y_i)-\mu_Y}{\sigma_Y}\Big)
\\&+ f_{0,\mathcal{S}}g_{1,\mathcal{S}}\prod_{i\in \mathcal{S}}\Big(\frac{\mathbb{E}(X_i)-\mu_X}{\sigma_X}\frac{\mathbb{E}(Y_i)-\mu_Y}{\sigma_Y}\Big)\mathbb{E}(Z) + f_{1,\mathcal{S}}g_{0,\mathcal{S}}\prod_{i\in \mathcal{S}}\Big(\frac{\mathbb{E}(X_i)-\mu_X}{\sigma_X}\frac{\mathbb{E}(Y_i)-\mu_Y}{\sigma_Y}\Big)\mathbb{E}(Z)\nonumber 
\\&+ f_{1,\mathcal{S}}g_{1,\mathcal{S}}\prod_{i\in \mathcal{S}}\Big(\frac{\mathbb{E}(X_i)-\mu_X}{\sigma_X}\frac{\mathbb{E}(Y_i)-\mu_Y}{\sigma_Y}\Big)\mathbb{E}(Z^2)\Bigg) = f_{0,\phi}g_{0,\phi}+ f_{1,\phi}g_{1,\phi}.
\label{eq:Prop1:2}
\end{align}
On the other hand:
\begin{align}
    &\mathbb{E}(U_d)=2P(U_d=1)-1,\quad  \mathbb{E}(V_d)=2P(V_d=1)-1,\label{eq:Prop1:3}
    \\&\qquad \qquad \mathbb{E}(U_dV_d)= 2P(U_d=V_d)-1. \label{eq:Prop1:4}
\end{align}
Let us denote $a=P(U_d=1)$ and $b=P(V_d=1)$. 
Combining Equations \eqref{eq:Prop1:1}-\eqref{eq:Prop1:4}, we get:
\begin{align}
&f_{0,\phi}=2a-1,\quad g_{0,\phi}=2b-1\nonumber\\
&(2a-1)(2b-1)+f_{1,\phi}g_{1,\phi}= 2P(U_d=V_d)-1.\label{eq:Prop1:4.5}
\end{align}
Next, we note that $f_d$ and $g_d$ are Boolean functions which take values in $\{-1,1\}$. As a result, $f_d(Z,X^d), g_d(Z,Y^d)\in [-1,1]$. We get: 
\begin{align*}
&-1 \le {f}_{0,\phi} + {f}_{1,\phi} \le 1, \qquad -1 \le {f}_{0,\phi} - {f}_{1,\phi} \le 1, \\
&-1 \le {g}_{0,\phi} + {g}_{1,\phi} \le 1, \qquad -1 \le {g}_{0,\phi} - {g}_{1,\phi} \le 1.
\end{align*}
Replacing ${f}_{0,\phi}=2a-1$ and ${g}_{0,\phi}=2b-1$, we get:
\begin{align*}
       &-2\min\{a, 1 - a\} \leq {f}_{1,\phi} \leq 2\min\{a,1 - a\}
    \\& -2\min\{b, 1 - b\} \leq {g}_{1,\phi} \leq 2\min\{b,1 - b\} 
\end{align*}
Consequently, from Equation \eqref{eq:Prop1:4.5}, we have:
\begin{align*}
    \zeta_{a,b} - 2\beta_{ab} \leq P(U_d=V_d) \leq \zeta_{a,b} + 2\beta_{ab} ,
\end{align*}
where we have defined $\zeta_{a,b}\triangleq 2ab-a-b+1$ and $\beta_{a,b}\triangleq \min\{a, (1 - a)\}\min\{b, (1 - b)\}$. This concludes the proof of $\mathcal{P}_{CR}(P_XP_Y,\mathcal{U},\mathcal{V})\subseteq \mathcal{P}'$. 
\\\textbf{Proof of $\mathcal{P}'\subseteq \mathcal{P}_{CR}(P_XP_Y,\mathcal{U},\mathcal{V})$}:
\\Let $Q_{U,V}$ be such that 
  \begin{align}
    \zeta_{a,b} - 2\beta_{ab} \leq P(U=V) \leq \zeta_{a,b} + 2\beta_{ab} ,\label{eq:Prop1:6}
\end{align}
where $a=Q_U(1)$ and $b=Q_V(1)$. Let $\tilde{f}_{1,\phi}$ and $\tilde{g}_{1,\phi}$ be real numbers satisfying:
\begin{align}
       &-2\min\{a, 1 - a\} \leq \tilde{f}_{1,\phi} \leq 2\min\{a,1 - a\}
    \\& -2\min\{b, 1 - b\} \leq \tilde{g}_{1,\phi} \leq 2\min\{b,1 - b\} 
    \\& \tilde{f}_{1,\phi}\tilde{g}_{1,\phi}= 2(P(U=V)-\zeta_{a,b}).\label{eq:Prop1:7}
\end{align}
Note such values for $\tilde{f}_{1,\phi}$ and $\tilde{g}_{1,\phi}$ always exist by the intermediate value theorem (IVT). For instance, we may fix $\tilde{f}_{1,\phi}=2\min\{a,1-a\}$ and change the value of $\tilde{g}_{1,\phi}$ within $[-2\min\{b, 1 - b\},2\min\{b, 1 - b\}]$ and apply the IVT along with the condition in Equation \eqref{eq:Prop1:6}.
Let us define 
\begin{align*}
    &\tilde{f}(Z)\triangleq (2a-1)+\tilde{f}_{1,\phi}Z,\qquad 
    \tilde{g}(Z)\triangleq (2b-1)+\tilde{g}_{1,\phi}Z.\\
\end{align*}
Then, using Equation \eqref{eq:Prop1:7}, we conclude $\tilde{f}(z), \tilde{g}(z)\in [-1,1]$ for all values of $z$. The source simulation scheme to simulate $P_{U,V}$ is as follows. Alice observes $Z$ and uses its local randomness to generate a Boolean variable $U$ where $P(U=1) =\frac{1+\tilde{f}(Z)}{2}$ and  $P(U=-1) =\frac{1-\tilde{f}(Z)}{2}$. Note that this is a valid distribution since $\tilde{f}(Z)\in [-1,1]$. Such a Boolean variable can always be produced using local randomness, e.g., using von Neumann's method \cite{VonNeumann1951}. Similarly, Bob locally generates $V$ such that $P(V=1)=\frac{1+\tilde{g}(Z)}{2}$ and $P(V=-1)=\frac{1-\tilde{g}(Z)}{2}$. We claim that $(U,V)$ simulates $Q_{U,V}$. To prove this claim, since any distirbution on pairs of binary variables has three degrees of freedom, it suffices to show that $\mathbb{E}(U)$, $\mathbb{E}(V)$, and $\mathbb{E}(UV)$ for the simulated vairbales are equal to those induced by $Q_{U,V}$. Note that
\begin{align*}
    &P(U=1)= \frac{1+\mathbb{E}(U)}{2}= \frac{1+\frac{1+\mathbb{E}(\tilde{f}(Z))}{2}-\frac{1-\mathbb{E}(\tilde{f}(Z))}{2}}{2}=a\\
    &P(V=1)= \frac{1+\mathbb{E}(U)}{2}=\frac{1+\frac{1+\mathbb{E}(\tilde{g}(Z))}{2}-\frac{1-\mathbb{E}(\tilde{g}(Z))}{2}}{2}=b\\
    \\&P(U=V)= \frac{1+\mathbb{E}(UV)}{2},
\end{align*}
On the other hand:
\begin{align*}
    &\mathbb{E}(UV)= 2P(U=V)-1
    \\&= 2\mathbb{E}(\frac{1+\tilde{f}(Z)}{2}\frac{1+\tilde{g}(Z)}{2})
    +2\mathbb{E}(\frac{1-\tilde{f}(Z)}{2}\frac{1-\tilde{g}(Z)}{2})-1
    \\&= \mathbb{E}(\tilde{f}(Z)\tilde{g}(Z))=(2a-1)(2b-1)+\tilde{f}_{1,\phi}\tilde{g}_{1,\phi}
    \\&= 2\zeta_{a,b}-1+2(P(U=V)-\zeta_{a,b})=2P(U=V)-1. 
\end{align*}
This completes the proof. 
\qed
\section{Proof of Proposition \ref{th:2}}
\label{App:th:2}
Similar to the proof of Proposition \ref{th:1}, it suffices to show that $\mathcal{P}_{EA}(P_{X,Y},\mathcal{U},\mathcal{V},\{-1,1\})\subseteq \mathcal{P}_{CR}(P_{X,Y},\mathcal{U},\mathcal{V})$. To this end, let us consider an EA-NISS setup and 
fix $f_d,g_d$ and ${\Lambda}_{d,i,z}, i\in \{1,2\},z\in \{-1,1\}$. Let $Z_1$ and $Z_2$ be the classical measurement outcomes observed by Alice and Bob, respectively, and $U_d=f_d(Z_1,X^d), V_d=g_d(Z_2,Y^d)$ the simulated output variables. Let us define:
\begin{align*}
    U_{x^d}= f_d(Z_1,x^d),\quad V_{y^d}= f_d(Z_2,y^d),\quad  x^d,y^d\in \{-1,1\}^d.
\end{align*}
Note that the distribution $P_{U_d,V_d}$ has three free variables, and is completely characterized by $\mathbb{E}(U_d), \mathbb{E}(V_d)$ and $\mathbb{E}(U_d,V_d)$. Hence, in order to show that the distribution can be simulated classically, it suffices to show that a pair of variables $U',V'$ can be simulated classically such that:
\begin{align}
&  \mathbb{E}(U')=\sum_{x^d\in \{-1,1\}^d} P_{X^d}(x^d)\mathbb{E}(U_{x^d}),\label{eq:th2:1}\\
&     \mathbb{E}(V')= \sum_{y^d\in \{-1,1\}^d} P_{Y^d}(y^d)\mathbb{E}(U_{y^d}),\label{eq:th2:2}\\
&      \mathbb{E}(U'V')=
\sum_{x^d,y^d\in \{-1,1\}^d}P_{X^d,Y^d}(x^d,y^d)\mathbb{E}(U_{x^d}V_{y^d})\label{eq:th2:3}.
\end{align}
We will show the latter claim by constructing the simulating functions achieving the above equalities. We first need to define several intermediate functions which will be used in constructing the simulating functions as follows. 

Note that for any fixed $x^d,y^d\in \{-1,1\}^d$, $U_{x^d}$ and $V_{y^d}$ are functions of $Z_1$ and $Z_2$, respectively, and potentially local randomness (due to the fact that the processing functions are potentially stochastic functions). From Corollary \ref{cor:1}, we conclude that $P_{U_{x^d},V_{y^d}}$ can be simulated classically by Alice and Bob using one bit of common randomness along with local randomness for any $x^d,y^d\in \{-1,1\}^d$. Let the simulating functions be denoted by $f_{x^d,y^d},g_{x^dy^d}, x^d,y^d\in 
\{-1,1\}^d$, so that if we set $\widetilde{U}_d\triangleq f_{x^d,y^d}(Z,X_{d+1}^{d+d'})$ and $\widetilde{V}_d\triangleq g_{x^d,y^d}(Z,Y_{d+d'+1}^{d+2d'})$ for some fixed $d'\in \mathbb{N}$, then $P_{\widetilde{U}_d,\widetilde{V}_d}$ can be made arbitrarily close to $P_{U_d,V_d}$ by appropriate choice of the simulating functions and blocklength. Note that the inputs $X_{d+1}^{d+d'}$ and $Y_{d+d'+1}^{d+2d'}$ are chosen such that their indices do not overlap, so that they follow a product distribution and Proposition \ref{th:1} can be applied. Next, similar to the proof of Proposition \ref{th:1}, using the Boolean Fourier decomposition, we can write:
\begin{align*}
   & f_{x^d,y^d}(Z,X_{d+1}^{d+d'})\!\!=\!\! f_{x^d,y^d,0,\phi}+Zf_{x^d,y^d,1,\phi}
+\sum_{\mathcal{S}\subseteq [d+1,d+d'],\mathcal{S}\neq \phi} \!\!\!\!\!\!\!\!\!(f_{x^d,y^d,0,\mathcal{S}}\!+\!Zf_{x^d,y^d,1,\mathcal{S}})\prod_{i\in \mathcal{S}} \frac{X_i\!-\!\mu_X}{\sigma_X},
    \\&g_{x^d,y^d}(Z,Y_{d+d'+1}^{d+2d'})\!\!= \!\!g_{x^d,y^d,0,\phi}+Zg_{x^d,y^d,1,\phi}+\sum_{\mathcal{S}\subseteq [d+d'+1,d+2d'],\mathcal{S}\neq \phi} \!\!\!\!\!\!\!\!\!\!\!\!(g_{x^d,y^d,0,\mathcal{S}}\!+\!Zg_{x^d,y^d,1,\mathcal{S}})\!\prod_{i\in \mathcal{S}} \frac{Y_i\!-\!\mu_Y}{\sigma_Y}.
\end{align*}
Let us define the following functions:
\begin{align*}
  &  \tilde{f}_{x^d,y^d}(Z)= f_{x^d,y^d,0,\phi}+f_{x^d,y^d,1,\phi}Z,
   \\& \tilde{g}_{x^d,y^d}(Z)= g_{x^d,y^d,0,\phi}+g_{x^d,y^d,1,\phi}Z,
\end{align*}
where $x^d,y^d\in \{-1,1\}^d$. Note that $\tilde{f}_{x^d}(\cdot),\tilde{g}_{y^d}(\cdot)\in [-1,1]$, to see this, we have the following:
\begin{align*}
    \sum_{x_{d+1}^{d+d'}}\!\!f_{x^d,y^d}(\!Z,x_{d+1}^{d+d'})\!=\! 2^{d'}\!\!(f_{x^d,y^d,0,\phi}\!+\!f_{x^d,y^d,1,\phi}Z)\!\in\! [-2^{d'}\!\!,\!2^{d'}\!],
\end{align*}
where we have used the fact that $f_{x^d,y^d}(Z,x_{d+1}^{d+d'})\in\{-1,1\}$ for all input values. Next, let us define:
\begin{align*}
  &  \overline{f}_{x^d,y^d}(Z)\triangleq f_{x^d,y^d,0,\phi}+|f_{x^d,y^d,1,\phi}|Z,\quad 
   \\& \overline{g}_{x^d,y^d}(Z)\triangleq g_{x^d,y^d,0,\phi}+|g_{x^d,y^d,1,\phi}|Z.
\end{align*}
Since $\tilde{f}_{x^d,y^d}(z),\tilde{g}_{x^d,y^d}(z)\in [-1,1]$ for all values of $z$, we must have 
$\overline{f}_{x^d,y^d}(z),\overline{g}_{x^d,y^d}(z)\in [-1,1]$. 
Note that $\mathbb{E}(U_{x^d})= f_{x^d,y^d,0,\phi}$ for all $y^d\in \{-1,1\}^d$. So, we define $f_{x^d,0,\phi}\triangleq f_{x^d,y^d,0,\phi}$. Similarly, $g_{y^d,0,\phi}\triangleq g_{x^d,y^d,0,\phi}$.
Furthermore, we define:
\begin{align*}
&    f^+_{x^d}(Z)\triangleq f_{x^d,0,\phi}+f^+_{x^d,1,\phi}Z,\quad  g^+_{y^d}(Z)= g_{y^d,0,\phi}+g^+_{y^d,1,\phi}Z
\\&    g^-_{y^d}(Z)= g_{y^d,0,\phi}+g^-_{y^d,1,\phi}Z,
\end{align*}
where 
\begin{align*}
   & f^+_{x^d,1,\phi}\triangleq\max_{y^d\in \{-1,1\}^d}|f_{x^d,y^d,1,\phi}|,
   \\& g^+_{y^d,1,\phi}\triangleq\max_{x^d\in \{-1,1\}^d}|g_{x^d,y^d,1,\phi}|,
   \\& g^-_{y^d,1,\phi}\triangleq-g^+_{y^d,1,\phi}.
\end{align*}
Note that $f^+_{x^d}(\cdot),g^+_{y^d}(\cdot),g^-_{x^d}(\cdot)\in [-1,1]$ since $\overline{f}_{x^d,y^d}, \overline{g}_{x^d,y^d}\in [-1,1]$ for all $x^d,y^d$ and all input values.  Lastly, let us define:
\begin{align*}
    &a\triangleq\sum_{x^d\in \{-1,1\}^d}P_{X^d}(x^d)\mathbb{E}(f^+_{x^d}(Z)),
    \\& b\triangleq\sum_{y^d\in \{-1,1\}^d}P_{Y^d}(y^d)\mathbb{E}(g^+_{y^d}(Z)),
    \\& \rho^+= \sum_{x^d,y^d\in \{-1,1\}^d}P_{X^d,Y^d}(x^d,y^d)\mathbb{E}(f^+_{x^d}(Z)g^+_{y^d}(Z)),
    \\& \rho^-=\sum_{x^d,y^d\in \{-1,1\}^d}P_{X^d,Y^d}(x^d,y^d)\mathbb{E}(f^+_{x^d}(Z)g^-_{y^d}(Z)).
\end{align*}
Note that by construction, for all $x^d,y^d\in \{-1,1\}^d$, we have:
\begin{align*}
 &   \mathbb{E}(f^+_{x^d}(Z)g^+_{y^d}(Z))= f_{x^d,0,\phi}g_{y^d,0,\phi}+ f^+_{x^d,1,\phi}g^+_{y^d,1,\phi}
   \geq  f_{x^d,0,\phi}g_{y^d,0,\phi}+ |f_{x^d,y^d,1,\phi}g_{x^d,y^d,1,\phi}|
\\&\geq f_{x^d,0,\phi}g_{y^d,0,\phi}+ f_{x^d,y^d,1,\phi}g_{x^d,y^d,1,\phi}= \mathbb{E}(U_{x^d}V_{y^d}).
\end{align*}
Similarly, 
\begin{align*}
&    \mathbb{E}(f^+_{x^d}(Z)g^-_{y^d}(Z))\leq  \mathbb{E}(U_{x^d}V_{y^d}).\\
\end{align*}
Consequently, $\rho^-\leq \rho'\leq \rho^+$, where \[\rho'\triangleq \sum_{x^d,y^d\in \{-1,1\}^d}P_{X^d,Y^d}(x^d,y^d)\mathbb{E}(U_{x^d}V_{y^d}).\]
Let $p_{ts}\triangleq \frac{\rho'-\rho^-}{\rho^+-\rho^-}$. Note that $p_{ts}\in [0,1]$ since  $\rho^-\leq \rho'\leq \rho^+$. 

The CR-NISS simulating scheme is as follows. Alice observes $X^d=x^d$, and computes $f^+_{x^d}(Z)$. Similar to the proof of Proposition \ref{th:1}, she uses independent copies of $X_i, i>d$ to generate a binary variable $U'$ such that $P_{U'}(1)= \frac{1+f^+_{x^d}(Z)}{2}$ and
$P_{U'}(-1)= \frac{1-f^+_{x^d}(Z)}{2}$. Alice outputs $U'$ as its simulated variable. Bob first generates a binary random variable $T$ such that $P_T(1)= p_{ts}$ and $P_T(-1)=1-p_{ts}$, using copies of $Y_i, i>d$ which are non-overlapping with those used by Alice.  Bob then observes $Y^d=y^d$, and uses additional non-overlapping copies of $Y_i,i>d$ to generate a binary random variable $V^+$ such that $P_{V^+}(1)= \frac{1+g^+_{y^d}(Z)}{2}$ and
$P_{V^+}(-1)= \frac{1-g^+_{y^d}(Z)}{2}$, and a binary random variable $V^-$ such that  $P_{V^-}(1)= \frac{1+g^-_{y^d}(Z)}{2}$ and
$P_{V^-}(-1)= \frac{1-g^-_{y^d}(Z)}{2}$. If $T=1$, Bob outputs $V'=V^+$ and if $T=-1$, Bob outputs $V'=V^-$. 

It remains to verify that Equations \eqref{eq:th2:1}-\eqref{eq:th2:3} are satisfied. First, note that by construction, we have:
\begin{align*}
    &\mathbb{E}(f^+_{x^d}(Z))= f_{x^d,\phi}=\mathbb{E}(U_{x^d}), \quad x^d\in \{-1,1\}^d
    \Rightarrow \mathbb{E}(U')=\sum_{x^d\in \{-1,1\}^d} P_{X^d}(x^d)\mathbb{E}(U_{x^d}),
    \end{align*}
    So, Equation \eqref{eq:th2:1} is satisfied. Similarly, 
    \begin{align*}
        \mathbb{E}(V')=P_T(1)\mathbb{E}(V^+)+P_T(-1)\mathbb{E}(V^-),
    \end{align*}
    and $\mathbb{E}(V^+)=\mathbb{E}(V^-)=\sum_{y^d\in \{-1,1\}^d}P_{Y^d}(y^d)g_{y^d,\phi}$. So, 
    \begin{align*}
        \mathbb{E}(V')&\!\!= \!p_{ts}\!\!\sum_{y^d}\!\!P_{Y^d}(y^d)g_{y^d,\phi}\!\!+\!\!(1\!\!-\!\!p_{ts})\!\!\sum_{y^d}\!\!P_{Y^d}(y^d)g_{y^d,\phi}
        \\& =\sum_{y^d\in \{-1,1\}^d}P_{Y^d}(y^d)\mathbb{E}(V_{y^d}).
    \end{align*}
    So, Equation \eqref{eq:th2:2} is satisfied. Lastly, 
    \begin{align*}
        \mathbb{E}(U'V')&= P_T(1)\mathbb{E}(U'V^+)+P_T(-1)\mathbb{E}(U'V^-)
        \\&= p_{ts}\rho^++(1-p_{ts})\rho^-= \rho',
    \end{align*}
    where we have used the fact that by defintion $p_{ts}= \frac{\rho'-\rho^-}{\rho^+-\rho^-}$. So, Equation \eqref{eq:th2:3} is satisfied. This concludes the proof. \qed

\end{appendices}

\end{document}